\documentclass[conference,letterpaper]{IEEEtran}
\IEEEoverridecommandlockouts

\usepackage{cite}
\usepackage{amsmath,amssymb,amsfonts}
\usepackage{mathrsfs}
\usepackage{graphicx}
\usepackage{textcomp}
\usepackage{xcolor}
\def\BibTeX{{\rm B\kern-.05em{\sc i\kern-.025em b}\kern-.08em
    T\kern-.1667em\lower.7ex\hbox{E}\kern-.125emX}}

\usepackage{graphicx} % Required for inserting images

\newcommand{\grad}{\nabla}
\usepackage{graphicx} % Required for inserting images
\usepackage{float}
\usepackage{cancel}
\graphicspath{{plots/}}
\usepackage{amsmath, amssymb, amsfonts,amsthm,comment,enumitem}
\usepackage{mathtools}
\usepackage[dvipsnames, svgnames]{xcolor}

\usepackage{algorithm}
\usepackage{algpseudocode}

\newtheorem{theorem}{Theorem}
\newtheorem*{theorem*}{Theorem}
\newtheorem{corollary}{Corollary}
\newtheorem{lemma}{Lemma}
\newtheorem{defn}{Defintion}
\newtheorem{proposition}{Proposition}
\newtheorem{remark}{Remark}
\newtheorem{fact}{Fact}

\newcommand{\poubelle}[1]{}

\title{PDE based Bounds on Optimal Exponents in Active Simple Hypothesis Testing}
\date{January 2026}
\author{\IEEEauthorblockN{Himadri Mandal$^\dagger$}\thanks{$^{\dagger}$Work done while at Ashoka University.}

\IEEEauthorblockA{
\textit{Indian Statistical Institute, Kolkata}\\
\textit{SCDLDS, Ashoka University} \\
bs2327@isical.ac.in \\}
\and
\IEEEauthorblockN{Sushant Vijayan}
\IEEEauthorblockA{
\textit{Plaksha University}\\
Mohali, India \\
sushant.vijayan@plaksha.edu.in} 
}
\begin{document}

\maketitle

\begin{abstract}
We present general upper and lower bounds on the optimal error exponent for the Active Simple Hypothesis Testing (ASHT) problem. Although ASHT has been extensively studied across diverse fields—including Information Theory, Bandits, and Operations Research—a precise characterization of this optimal exponent has remained elusive. This work combines a recent Partial Differential Equation (PDE) characterization of the optimal exponent by \cite{vijayan2025minimaxoptimalalgorithmsactive} with the comparison principle of viscosity solutions of Hamilton-Jacobi PDEs. This leads to bounds that are parametrized by convex functions, sets and lower Hamiltonian relaxations, a flexibility that generalizes many known results in the literature.

%Note: THIS PAPER IS ELIGIBLE FOR THE STUDENT PAPER AWARD.
\end{abstract}


% \section{Final Format: \rred{TEMP SECTION}}
% \begin{itemize}
%     \item Introduction 
%     \begin{itemize}
%         \item See the "controlled sensing in multihypothesis testing" By veeravalli et al. This is the exact problem we are solving. So we will situate the problem based on what they have said. We will also look at newer papers that have cited it. Just motivate the problem broadly with applications and prior work.
%     \end{itemize}
%     \item Literature Review
%     \begin{itemize}
%         \item We will write a couple of pars citing work from bandits, info theory. We will spend some more time on Komiyama and my paper to describe the fundamental limit
%     \end{itemize}
%     \item Problem Formulation
%     \begin{itemize}
%         \item Clear mathematical formulation of the problem. See my paper and veeravalli paper to write this mathematically.
%     \end{itemize}
%     \item{Main Results}
%     \begin{itemize}
%         Describe the PDE version of the optimal exponent from my paper and the results from Abgrall paper. Present the calculations you have done based on Abgrall.
%     \end{itemize}
%     \item Conclusion
%     \begin{itemize}
%         \item We discuss why our bounds improve over Veeravalli et al and why is it more tractable over $R^{go}_\infty$. 
%     \end{itemize}
% \end{itemize}

\section{Introduction}
 In Active Simple Hypothesis Testing (ASHT), a controller chooses an action $a \in [K]$, and receives a sample observation conditioned on $a$ from the underlying unknown hypothesis $H_i$. We assume that $H_i$ is from a known family of channels $\mathcal{H}$. The controller can adaptively choose its future actions based on the prior history of observations and actions. After the budget of $T$ samples has been exhausted, it is required to output a hypothesis $\hat{i}_T \in \mathcal{H}$. The controller aims to minimize probability of error $P_e$ it incurs when choosing $\hat{i}_T$. For a fixed budget $T$, an optimal solution may be obtained through dynamic programming, but it is quite inefficient computationally and uninsightful. Further, it is well-known from prior works such as \cite{audibert2010best,komiyama2022minimax} that the error probability decays exponentially in the sample budget $T$. Hence, a natural goal is to study the optimal error exponent in the rate of decay of $P_e$ as $T \to \infty$. This work, therefore, studies the  ASHT problem in a large deviations setting \cite{dembo2009large}.
 
This problem is motivated by a wide range of applications, including sensor management \cite{castanon1997approximate,atia2011sensor}, channel identification \cite{hayashi2009discrimination}, and sequential planning in clinical trials \cite{armitage1960sequential}. Owing to its fundamental nature, ASHT has been independently studied across multiple disciplines under various names: as Channel Discrimination \cite{hayashi2009discrimination} and Controlled Sensing \cite{nitinawarat2013controlled} in Information Theory; Ranking and Selection \cite{glynn2004large} in Operations Research; and Fixed-Budget Best Arm Identification \cite{audibert2010best,komiyama2022minimax} in the bandit literature. 

In this work, we develop tighter bounds on the optimal error exponent achievable by any feasible control strategy compared to prior works like \cite{glynn2004large,hayashi2009discrimination,nitinawarat2013controlled}. We achieve this by combining a characterization \cite{vijayan2025minimaxoptimalalgorithmsactive} of the optimal exponent in terms of a Partial Differential Equation (PDE) and stability bounds \cite{abgrall1994some} for such Hamilton-Jacobi-Isaac (HJI) PDEs. %We rewrite the terminal function $g(x)=x_{(2)}$ \footnote{$x_{(j)}$ is the $j^{th}$ order statistic of the vector $x$.} of the PDE into a sum of concave and convex functions and apply the corresponding Hopf formulae (\cite{bardi1984hopf}) to obtain our bounds.
\section{Problem Formulation}
%\rred{Define the PROBLEM well here!}
We study a hypothesis testing problem with $m$ hypotheses, where the set of hypotheses is denoted by $\mathcal{H} = \{H_1, \cdots, H_{m}\}$. 
Each hypothesis $H_i$ is assumed to be a single channel, that is:
$
H_i := \{P_a^i \ | \ a \in \mathcal{A}\}.
$
Here $\mathcal{A}$ denotes the common action set of all the hypotheses. At each time step, the agent chooses an action $a \in \mathcal{A}$ (with $|\mathcal{A}|=K$), and ``nature" returns samples from the conditional distribution of action $a$ according to the true unknown underlying channel. We assume these conditional distributions have full support over the common finite set $\mathcal{Y}$. 

The controller has a sample budget $T$, and the controller's policy $\pi_T$ consists of $T$ many \textit{sampling rules}, denoted as $\pi_T^{t}$, for each $ t \in [T]$ and a final \textit{decision rule} $\hat{i}_T$. The sampling rule is \textit{causal} and hence, each $\pi^{t}_T$ is adapted to the sample filtration $\mathcal{F}_t$. The decision rule $\hat{i}_T \in [m]$ is adapted to the final sample filtration $\mathcal{F}_T$ and is interpreted as the final hypothesis choice of the controller.
We denote the set of all such admissible causal controller policies for a horizon $T$ as $\boldsymbol{\Pi}_T$. $\boldsymbol{\Pi}$ is the infinite sequence $(\pi_1,\dots,\pi_T,\dots)$ of such policies such that $\pi_T \in \boldsymbol{\Pi}_T, \forall T \in \mathbb{N}$. 
\begin{defn}
The \textbf{error probability} of a policy $\pi_T$ on the channel hypothesis $H_i \in \mathcal{H}$ to be:
$$
P_e(i, \pi_T) := P_{H_i, \pi_T}(\hat{i}_T \neq i).
$$ 

    The \textbf{error exponent} of an admissible sequence of controller policies $\Pi \in \boldsymbol{\Pi}$ for $H_i$, is defined as
    $$
    e(H_i, \Pi) = \liminf_{T \to \infty} -\frac{\log(P_e(i, \pi_T))}{T}
    $$
    and the \textbf{minimal} exponent over the hypotheses $\mathcal{H}$ as
    $$
    e_m(\mathcal{H}, \Pi) = \min_{H_i \in \mathcal{H}} e(H_i, \Pi).
    $$
    The \textbf{minimax-optimal} exponent is then defined as:
    $$
    e^{*}(\mathcal{H})=\underset{\Pi \in \boldsymbol{\Pi}}{\max} \hspace{0.2cm} e_m(\mathcal{H}, \Pi),
    $$
        and a policy sequence $\Pi^{*}$ that attains this maxima\footnote{A dynamic programming argument shows the existence of $\Pi^{*}$.} is called the \textbf{minimax-optimal policy}.
\end{defn}

Our goal is to find tighter upper and lower bounds for $e^{*}(\mathcal{H})$.

%The problem of finding the minimax optimal policy, as studied in \cite{vijayan2025minimaxoptimalalgorithmsactive}, is still a challenging open problem. This is because there is very little theory on Hamilton-Jacobi type PDEs with nonconvex-nonconcave Hamiltonian 

\section{Related Work}

\cite{hayashi2009discrimination} showed that for the binary testing ($m=2$) case, the non-adaptive strategy of playing the action that maximally discriminates the hypotheses in Chernoff Information is optimal. \cite{nitinawarat2013controlled} obtained, for $m>2$, the upper bound $e^{*}(\mathcal{H}) \leq \min_{i \neq j} \max_a C(P^{i}_a,P^{j}_a)$, where $C(P,Q)$ denotes the Chernoff Information between $P$ and $Q$. They present an example where adaptive strategies outperform non-adaptive ones when $m=3$. \cite{komiyama2022minimax} used an adaptive sequence of data processing inequalities to obtain a family of upper bounds, which they denoted as $R^{go}_B$, $B \in \mathbb{N}$. Further, they showed that $\lim_{B \to \infty} R^{go}_B=e^{*}(\mathcal{H})$. They also proposed an optimal policy that was intractable computationally. To side-step this difficulty, they construct a neural network-based algorithm that has good empirical performance but comes with no optimality guarantees. 

\cite{vijayan2025minimaxoptimalalgorithmsactive} showed that $e^{*}(\mathcal{H})=V(0,1)$, where the function $V$ is the \textit{viscosity solution} of the Hamilton-Jacobi-Isaacs (HJI) PDE:
\begin{equation}\label{hji_pde}
\partial_t V + H(\grad V) = 0
\end{equation}
with the initial condition $V(x,0) = x_{(2)}$ \footnote{$x_{(i)}$ is the $i^{th}$ order statistic of the vector $x$.}. The Hamiltonian $H$ has the form $H=\min_a H_a$. Here $H_a$ is defined as:
\begin{equation*}\label{Hamiltonian_formula}
H_a(p)=
\begin{cases}
     P \log \left( \sum_{y \in \mathcal{Y}}\prod^{m}_{i=1}(P^{i}_a(y))^{\frac{p_i}{P}} \right), & \text{ if } P>0\\
    \max_{y \in \mathcal{Y}}\sum_i p_i \log(P^{i}_a(y)), & \text { if } P \leq 0
\end{cases}
\end{equation*}
with $P=\sum_i p_i$. This HJI-PDE is a consequence of interpreting $e^{*}(\mathcal{H})$ as the value of a certain differential game between the controller and nature, wherein each player tries to optimize the likelihood of error over all sample paths. Further, they obtained error exponent upper bounds by combining the Comparison principle and Hopf formulae for Hamilton-Jacobi PDEs. 

In this paper, we build on \cite{vijayan2025minimaxoptimalalgorithmsactive}'s formulation and use techniques from viscosity solution theory of PDEs to obtain various bounds. Our bounds are achieved by either perturbing the initial condition $g$ or the Hamiltonian $H$ by convex and concave approximations. We obtain three families of bounds, one parametrized by any proper convex functions, the second by proper convex sets and the third by lower relaxations of the Hamiltonian. This improved flexibility generalizes prior bounds obtained in the literature (cf. \cite{hayashi2009discrimination,nitinawarat2013controlled,komiyama2022minimax,vijayan2025minimaxoptimalalgorithmsactive}).

\section{Preliminaries}

Consider a generic Cauchy Problem in $\mathbb{R}^m \times (0, \infty)$ of the form

\begin{equation}\label{cauchy_prob}
\phi_t + H(\grad\phi) = 0
\end{equation}
with an initial condition
$
\phi(x,0) = \phi^{0}(x)
$
where $H$ is a Lipschitz continuous function in $\mathbb{R}^{m}$ and 
$\phi^{0}$ is a uniformly Lipschitz function. It is known that under these regularity conditions, the above Cauchy problem admits a unique Lipschitz continuous \textit{viscosity solution} \cite{crandall1983viscosity}. 

The \textit{Comparison Principle} \cite{crandall1983viscosity} states that for any two Cauchy problems for HJI PDEs such that $H_a \geq H_b$ and $\phi^{0}_a \leq \phi^{0}_b$, we must have that $\phi_a(x,t) \leq \phi_b(x,t), \forall (x,t) \in \mathbb{R}^m \times (0, \infty)$.

% To do this, we also need the \textit{Hopf} and \textit{Hopf-Lax} formula \rred{HOPF and HOPF-LAX FORMULAE SUMMARIZED} (see \cite{evans2022partial}) that provides the exact solution when the initial condition $\phi^0$ or the Hamiltonian is convex or concave. For the convex initial condition case,
% $$
% \phi(x,t) = \sup_{v\in\mathbb{R}^m}[x\cdot v - \phi^{0*}(v) - tH(v)]
% $$
% And for the concave case, we get

% $$
% \phi(x,t) = \inf_{v\in\mathbb{R}^m}[x\cdot v - \phi^{0*}(-v) - tH(v)]
% $$
% where $*$ represents the Legendre-Fenchel conjugate. A similar formula holds for the convex or concave Hamiltonian.

The \textit{Hopf and Hopf--Lax} formulae \cite{evans2022partial} provide closed-form solutions to the HJI PDEs whenever either the initial datum \(\phi^0\) or the Hamiltonian \(H\) is convex or concave. For a \textbf{convex} initial function \(\phi^0\), the solution is given by the Hopf formula:
\[
\phi(x,t) = \sup_{v \in \mathbb{R}^m} \Bigl[ x \cdot v - \phi^{0*}(v) - t \, H(v) \Bigr].
\]
Similarly, when $\phi^{0}$ is \textbf{concave}, the solution is:
\[
\phi(x,t) = \inf_{v \in \mathbb{R}^m} \Bigl[ x \cdot v - \phi^{0*}(-v) - t \, H(v) \Bigr].
\]
Here \(f^{*}\) denotes the appropriate Legendre--Fenchel conjugate \cite{rockafellar1997convex} of the function $f$. In the case when Hamiltonian \(H\) is \textbf{convex}, the solution is given by the Hopf--Lax formula:
\[
\phi(x,t) = \inf_{y \in \mathbb{R}^m} \Bigl[ \phi^0(y) + t \, H^*\bigl( \tfrac{x-y}{t} \bigr) \Bigr],
\]
and for a \textbf{concave} Hamiltonian by:
\[
\phi(x,t) = \sup_{y \in \mathbb{R}^m} \Bigl[ \phi^0(y) - t \, (-H)^*\bigl( \tfrac{x-y}{t} \bigr) \Bigr].
\]

We aim to obtain upper and lower bounds on the viscosity solution of ~\eqref{hji_pde} by approximating $\phi^{(0)}$ or $H$ by concave or convex bounds and applying the comparison principle. The results of \cite{vijayan2025minimaxoptimalalgorithmsactive} then translate these bounds to $e^{*}(\mathcal{H})$. 

The results of  \cite{abgrall1994some} are very useful in this endeavour. However, the precise result \footnote{Proposition 1.1 in \cite{abgrall1994some}.} stated there is incorrect. We give a corrected version of this result below and provide a proof.

\begin{proposition}[Corrected result of \cite{abgrall1994some}]{\label{Thm1}}
For any convex function $f$, let $D_f = \{x \in \mathbb{R}^m , f^*(x) < +\infty\}$ and similarly for a concave function let it denote $D_f=\{x \in \mathbb{R}^m , f^*(x) > -\infty\}$. For any initial datum $\phi^{0}=\phi^{0}_1+\phi^{0}_2$, where $\phi^{0}_1$ is convex and $\phi^{0}_2$ is concave, we have that
    \[
\begin{aligned}
\inf_{v\in D_{\phi^0_2}} \sup_{w\in \mathbb{R}^m}
\Bigl\{ x\cdot w - &\phi_2^{0*}(v) - \phi_1^{0*}(v+w) - tH(w) \Bigr\} \\
&\ge \phi(x,t) \ge \\
\sup_{v\in D_{\phi^0_1}} \inf_{w\in \mathbb{R}^m}
\Bigl\{ x\cdot w - &\phi_1^{0*}(v) - \phi_2^{0*}(v-w) - tH(w) \Bigr\}
\end{aligned}
\]
\end{proposition}
\begin{proof}
Since $\phi^{0*}_1(v) = \sup_{x\in\mathbb{R}^m} [x\cdot v - \phi^0_1(x)]$, we have that
\begin{align*}
    \phi^0_1(x) &\geq x \cdot v - \phi^{0*}_1(v) \\
    \phi^0(x) &\geq x \cdot v + \phi^0_2(x) - \phi^{0*}_1(v) := \xi_v(x)
\end{align*}
Clearly, $\xi_v(x)$ is a concave function of $x$. The Hopf formula for a Cauchy Problem with initial datum as $\xi_v$ gives its corresponding viscosity solution $\xi_v(x,t)$ as:
$$
\xi_v(x,t) = \inf_{w \in \mathbb{R}^m} \Bigg\{ x\cdot w - \xi^*_v(-w) - t H(w) \Bigg\}.
$$
We observe that:
\begin{align*}
\xi^*_v(-w) &= \inf_{x \in \mathbb{R}^m} \Bigg\{+x\cdot w - \xi_v(x)\Bigg\} \\
            &= \inf_{x\in\mathbb{R}^m} \Bigg\{x \cdot (w-v) + \phi^{0*}_1(v) - \phi^0_2(x) \Bigg\} \\
            &= \inf_{x\in\mathbb{R}^m} \Bigg\{-x \cdot (v-w) + \phi^{0*}_1(v) - \phi^0_2(x) \Bigg\}.
\end{align*}
Thus, we have:
$$
\xi^*_v(-w) = 
\begin{cases}
+\infty, & v \not\in D_{\phi^0_1} \\
\phi^{0*}_2(v-w) + \phi^{0*}_1(v) & v \in D_{\phi^0_1}. \\
\end{cases}
$$
Plugging this back into the Hopf representation, we obtain, $\forall v \in D_{\phi^0_1}$ that:
$$
\xi_v(x,t) = \inf_{w\in\mathbb{R}^m} \Bigg\{x\cdot w - \phi^{0*}_2(v - w) - \phi^{0*}_1(v) - tH(w)\Bigg\}.
$$
Using the comparison principle we get, for any $v \in D_{\phi^0_1}$
$$
\phi(x,t) \geq \xi_v(x,t) = \inf_{w\in\mathbb{R}^m} \Bigg\{x\cdot w - \phi^{0*}_2(v - w) - \phi^{0*}_1(v) - tH(w)\Bigg\}, 
$$
and by supremizing over all feasible $v$ we get the lower bound:
$$
\phi(x,t) \geq \sup_{v\in D_{\phi^0_1}}\inf_{w \in \mathbb{R}^m} \Bigg\{x\cdot w - \phi^{0*}_1(v)  -\phi^{0*}_2(v-w) - tH(w)\Bigg\}.
$$
We repeat the arguments for $\phi^{0}_2$ to obtain an upper bound. Since $\phi^{0*}_2(w)=\inf_{x\in\mathbb{R}^m} [x\cdot w - \phi^0_2(x)]$, we get that:
\[ \phi^{0}(x) \leq x \cdot w+\phi^{0}_1(x)-\phi^{0*}_2(w):= \eta_w(x).\]
Again applying Hopf formula since $\eta_w$ is convex function gives us a representation for the viscosity solution $\eta_w(x,t)$:
$$
\eta_w(x,t) = \sup_{v \in \mathbb{R}^m} \Bigg\{ x\cdot v - \eta^*_w(v) - t H(v) \Bigg\}.
$$
Again, as above, we have that:
$$
\eta^*_w(v) = 
\begin{cases}
-\infty, & w \not\in D_{\phi^0_2} \\
\phi^{0*}_2(w) + \phi^{0*}_1(w-v) & w \in D_{\phi^0_2}. \\
\end{cases}
$$
Using the comparison principle and infimizing over all feasible $w$ gives:
$$
\phi(x,t) \leq \inf_{w \in D_{\phi^0_2}}\sup_{v \in \mathbb{R}^m} \Bigg\{x\cdot w - \phi^{0*}_2(w)  -\phi^{0*}_1(w-v) - tH(v)\Bigg\}.
$$
\end{proof}

\section{Bounds on the solution of ~\eqref{hji_pde}.}
\subsection{Bounds using Abgrall's analysis}
To apply Proposition~\ref{Thm1} to HJI-PDE \eqref{hji_pde}, we start by decomposing the initial condition $x_{(2)}$ as a sum of convex and concave functions.
\begin{lemma}\label{conv_conc_decomp}
    $$
    x_{(2)} = \sigma_{P}(x) - \sigma_{Q}(x)
    $$
where $\sigma_{P}, \sigma_{Q}$ represent the support functions\footnote{$\sigma_R(x):=\underset{y \in R}{\sup} \langle y,x \rangle$, see for example \cite{rockafellar1997convex}.} of the convex bodies 
$$
P = \operatorname{convhull}(\{\bold{1} - e_{i}\}_i), \ Q = \operatorname{convhull}(\{\bold{1} - e_i - e_j\}_{i,j : i \neq j})
$$ 
respectively. Here $\{e_i\}_{i \in [m]}$ is the standard basis in $\mathbb{R}^{m}$.
\end{lemma}
\begin{proof}
Since
    \begin{align*}
    x_{(2)}  &=\Bigg[\sum_i x_i - \min_i\{x_i\}\Bigg] + \Bigg[-\max_{|S| = n-2} \sum_{i \in S} x_i\Bigg] \\ 
            &= \sup_{\lambda \in P} \langle \lambda , x\rangle - \sup_{\mu \in Q} \langle \mu, x \rangle,
    \end{align*}
we have that:
    $$
    x_{(2)} = \sigma_{P}(x) - \sigma_{Q}(x).
    $$
\end{proof}

%\rred{Now, we give an equivalent characterisation of the same. Let $\sigma \in S_n$. Split $\mathbb{R}^m$ as $\cup_{\sigma \in S_n} \Omega_\sigma$ such that 
%$$
%\Omega_\sigma = \Bigg\{x \in \mathbb{R}^m \Bigg | \ x_{\sigma(1)} \geq x_{\sigma(2)} \geq \dots \geq x_{\sigma(n)}\Bigg\}
%$$}

\begin{fact}{\label{fact1}}
For a convex body $P$, the Legendre-Fenchel conjugate of its support function satisfies 
    $$\sigma_P^*(x) = \delta_P(x)$$
where $\delta_P(x) = 0$ for $x \in P$ and $+\infty$ otherwise. 
\end{fact}
See standard convex analysis texts like \cite{rockafellar2015convex}, \cite{hiriart1996convex} \label{convex-analysis-books} for a proof of this.
Using this sum of convex and concave functions decomposition, one obtains the following upper and lower bounds:
\begin{theorem}\label{Thm2}
\[
\begin{aligned}
\inf_{v\in Q} \sup_{u\in P}
\Bigl\{ x \cdot &(u - v) - tH(u-v) \Bigr\}\\
&\ge V(x,t) \ge \\
\sup_{v\in P} \inf_{u\in Q}
\Bigl\{ x\cdot &(v-u) - tH(v-u) \Bigr\}
\end{aligned}
\]
And, since, $e^{*}(\mathcal{H}) = V(0,1)$ satisfies
\[
\begin{aligned}
\inf_{v\in Q} \sup_{u\in P}
\Bigl\{ &-H(u-v) \Bigr\}\\
&\ge e^{*}(\mathcal{H}) \ge \\
\sup_{v\in P} \inf_{u\in Q}
\Bigl\{ &-H(v-u) \Bigr\}
\end{aligned}
\]

\end{theorem}
\begin{proof}
    This is a direct application of Proposition~\ref {Thm1}. Use Fact~\ref{fact1} with the decomposition 
    $
    x_{(2)} = \sigma_P(x) - \sigma_Q(x).
    $
\end{proof}
\begin{remark}
In the upper bound in Theorem~\ref{Thm2}, we can show that for any choice of $v \in Q$, there exists a choice of $u \in P$, so that $u-v \in \Delta_{m}$. If it was true for any choice of $v \in Q$, the maximum is always attained when $u-v \in \Delta_{m}$, then  Theorem~\ref{Thm2} improves on Proposition 4.11 of \cite{vijayan2025minimaxoptimalalgorithmsactive} \footnote{\cite{vijayan2025minimaxoptimalalgorithmsactive} already improves on \cite{hayashi2009discrimination,nitinawarat2013controlled}.}. We are currently unable to show this and leave it for future work.   
\end{remark}

Using Theorem~\ref{Thm1}, one can get another class of bounds. This is done by noting that the sum of convex and concave function decomposition is not unique. Let $\mathscr{C}(\mathbb{R}^m)$ be the set of \textit{all} convex functions on $\mathbb{R}^m$ and $\mathscr{C}_P(\mathbb{R}^m) (\subset \mathscr{C}(\mathbb{R}^m))$ be the subset of all proper, convex, and lower semi-continuous functions. To begin, let $F \in \mathscr{C}_P(\mathbb{R}^m)$. Consider the initial function decomposition:
$$
x_{(2)} = \tilde\phi_1^{0} + \tilde\phi_2^{0} = [\sigma_P(x) + F^{*}(x)] - [\sigma_Q(x) +F^{*}(x)]
$$. 

\begin{fact}\label{fact-infimal}
    Let $f,g \in \mathscr{C}(\mathbb{R}^m)$ be proper, convex and lower semi-continuous. Then
    $$
        [f + g]^{*} = \left(f^* \ \square \ g^*\right),
    $$
    where $\square$ denotes the infimal convolution:
    $$
    f \square g(x):= \inf \{ f(r)+g(s) | r+s=x \},
    $$
    between two convex functions $f$ and $g$.
\end{fact}

Again, the reader should look at standard convex analysis books like~\cite{rockafellar2015convex} for a proof of the above fact.

\begin{theorem}\label{Thm3}
Let $F$ be a proper, convex and lower semi-continuous function. Then,
\[
\begin{aligned}
\inf_{v\in D_{\tilde{\phi}^0_2}} \sup_{z\in \mathbb{R}^m}
\Bigl\{(\delta_Q \ \square \ F)(&v) - (\delta_P \ \square \ F)(z) - H(z-v) \Bigr\} \\
&\ge e^{*}(\mathcal{H}) \ge \\
\sup_{v\in D_{\tilde{\phi}^0_1}} \inf_{z\in \mathbb{R}^m}
\Bigl\{-(\delta_P \ \square \ F)(&v) + (\delta_Q \ \square \ F)(z) - H(z-v) \Bigr\}.
\end{aligned}
\]
\end{theorem}
\begin{proof}
We have
$$
x_{(2)} = [\sigma_P(x) + F^{*}(x)] - [\sigma_Q(x) +F^{*}(x)] = \tilde\phi_1^{0} + \tilde\phi_2^{0}.
$$
By Fact~\ref{fact-infimal},

$$
{\tilde{\phi}_1}^{0*}(z) = [\sigma_P + F^{*}]^{*}(z) = (\delta_P \ \square \ F)(z)
$$

and similarly 

$$
 -{\tilde{\phi}_2}^{0*}(v) = (-\tilde\phi_2^0)^*(v) = [\sigma_Q + F^{*}]^{*}(v) = (\delta_Q \ \square \ F)(v).
$$

Now, applying Theorem~\ref{Thm1} one gets the result. 

\end{proof}
We also have the following result by optimizing the above over $F \in \mathscr{C}_P(\mathbb{R}^m)$:
\begin{corollary} \label{corr2}
\[
\begin{aligned}
\inf_{\substack{F \in \mathscr{C}_P(\mathbb{R}^m) \\ v\in D_{\tilde{\phi}^0_2}}}\sup_{z\in \mathbb{R}^m}
\Bigl\{(\delta_Q\square F)(&v) - (\delta_P\square F)(z) - H(z-v) \Bigr\} \\
&\ge e^{*}(\mathcal{H}) \ge \\
\sup_{\substack{F \in \mathscr{C}_P(\mathbb{R}^m) \\ v \in D_{\tilde{\phi}^0_1}}}\inf_{z\in \mathbb{R}^m}
\Bigl\{-(\delta_P\square F)(&v) + (\delta_Q \square F)(z) - H(z-v) \Bigr\}
\end{aligned}
\]
\end{corollary}
The bounds established have a distinct game-theoretic character to them. In the lower bounds, one wishes to make the supremizing player have more choices while lessening the choices of the infimizing player. The opposite considerations apply for the upper bound. In Theorem~\ref{Thm2}, the choices are hard constrained to the sets $P, Q$, while in the results Theorem~\ref{Thm3} and Corollary~\ref{corr2}, the choices are only soft constrained by distances to $P,Q$ as measured by $F$. 

We will now derive the bounds for some specific functions $F \in \mathscr{C}_P(\mathbb{R}^m)$. We study the simple classes of support functions and squared Euclidean distances. 
\begin{corollary}
    For $R$ a convex body take $F = \sigma_R$, 
    \begin{align*}
    \inf_{v\in D_{\tilde{\phi}^0_2}} \sup_{z\in \mathbb{R}^m}
    \Bigl\{(\delta_Q \ \square \ \sigma_R)(&v) - (\delta_P \ \square \ \sigma_R)(z) - H(z-v) \Bigr\} \\
    &\ge e^{*}(\mathcal{H}) \ge \\
    \sup_{v\in D_{\tilde{\phi}^0_1}} \inf_{z\in \mathbb{R}^m}
    \Bigl\{-(\delta_P \ \square \ \sigma_R)(&v) + (\delta_Q \ \square \ \sigma_R)(z) - H(z-v) \Bigr\}
    \end{align*}
\end{corollary}
This gives a bound based on arbitrary convex sets $R$.
\begin{corollary}\label{euclid_corollary}
    For $F(x) = \|x - c\|_2^2$,
    \begin{align*}
    \inf_{v\in D_{\tilde{\phi}^0_2}} \sup_{z\in \mathbb{R}^m}
    \Bigl\{\operatorname{d^2}(v, \ &Q+c) - \operatorname{d^2}(z, \ P+c) - H(z-v) \Bigr\} \\
    &\ge e^{*}(\mathcal{H}) \ge \\
    \sup_{v\in D_{\tilde{\phi}^0_1}} \inf_{z\in \mathbb{R}^m}
    \Bigl\{-\operatorname{d^2}(v, \ &P+c) + \operatorname{d^2}(z, \ Q+c) - H(z-v) \Bigr\}
    \end{align*}
    where $\operatorname{d^2}(x,C)$ denotes the euclidean distance of $x$ from a set $C$.
\end{corollary}

In Corollary \ref{euclid_corollary}, choosing the distance function penalizes the players from moving too far away from the sets $P$ and $Q$, but the penalization is not as strong as $F \equiv 0$. The hope is to find a sweet spot. 

One can write $-H(p) = \max_a (-H_a(p))$, switch $\inf, \max$ and obtain the exact infima using the fact that 
$$
\grad_v \operatorname{d^2}(v, Q+c) = 2(v - \pi_{Q+c}(v))
$$
where $\pi_R(v) = \arg\min_{q \in R} \|v - q\|_2^2$. 

Another natural choice for $F$ is $\delta_R$. However, one can easily prove, for that choice of $F$, the bound obtained would be necessarily worse than Theorem~\ref{Thm2}. We expect the tightest bounds to come from when $F$ is chosen to respect the natural geometry of the ASHT problem. One such choice is  $F(.)=\sum_a w_a H_a(.)$, with $w \in \Delta_{m}$. This leads to considerations involving I-projections \cite{nielsen2018information} on the sets $P,Q$. We intend to pursue this approach in future work.

\subsection{Bounds using a Hamiltonian relaxation}
In this section, we will obtain bounds by approximating the Hamiltonian. To this end, let $w = (w_a)_{a \in \mathcal{A}} \in \Delta_K$ i.e. $\sum_{a \in A} w_a = 1, w_a \geq 0$. Define $\tilde{H}_w(p) = \sum_{a \in A} w_a H_a(p)$. Since $H_a(\cdot)$ is convex, so is $\tilde{H}_w(\cdot)$. We then have that:
\begin{equation}\label{h_bounds}
 H^{**} \leq H \leq \tilde{H}_w.   
\end{equation}
We use these inequalities to bound $e^{*}(\mathcal{H})$ using the Hopf-Lax formula. 

\begin{theorem}[Hamiltonian Relaxation]
We have that:
    \[
    \begin{aligned}
    \min_{y\in\mathbb{R}^m} \sup_{w \in \Delta_K}  & \bigg \{t \tilde{H}_w^{*}\left(\frac{x-y}{t}\right) + y_{(2)} \bigg\}  \\    
     & \ge V(x,t) \ge \\
    \sup_{w \in \Delta_K}\min_{y\in\mathbb{R}^m} & \bigg \{t\tilde{H}_w^{*}\left(\frac{x-y}{t}\right) + y_{(2)} \bigg\} 
    \end{aligned}
    \]
    And, in particular, 
    \[ 
    \begin{aligned}
    \min_{y\in\mathbb{R}^m} \sup_{w \in \Delta_K} &  \bigg \{ \tilde{H}^{*}_w\left(-y\right) + y_{(2)} \bigg\}\\
   & \geq  e^{*}(\mathcal{H}) \geq \\
   \sup_{w \in \Delta_K} \min_{y \in \mathbb{R}^m} & \bigg \{\tilde{H}^{*}_w(- y) + y_{(2)} \bigg\}
    \end{aligned}
    \]
\end{theorem}

\begin{proof}
Consider a new Cauchy problem~\eqref{cauchy_prob} with the relaxed Hamiltonian $\tilde{H}_w(\cdot)$ and the same initial datum $y_{(2)}$. Denote the viscosity solution of this problem as $\tilde{V}_w$ , that is,

    \begin{equation}\label{pde-hamiltonian}
        \partial_t \tilde{V}_w + \tilde{H}_w(\grad \tilde{V}_w) = 0
    \end{equation}
with $\tilde{V}_w(x,0) = x_{(2)}$.

Applying the Comparison Principle with the inequality $H \leq \tilde{H}_w$ and then invoking Hopf-Lax, we then get:
    $$
    V(x, t) \geq \tilde{V}_w(x,t) = \min_{y \in \mathbb{R}^m} \bigg \{t\tilde{H}_w^{*}\left(\frac{x-y}{t}\right) + y_{(2)} \bigg \}
    $$
    
    Supremizing over the allocation $w = (w_a)_{a \in \mathcal{A}}$ we get
    \[
    V(x,t) \geq \sup_{w \in \Delta_K}\min_{y\in\mathbb{R}^m} \bigg \{t\tilde{H}_w^{*}\left(\frac{x-y}{t}\right) + y_{(2)} \bigg\}
    \]
    which proves the lower bound. For the upper bound, we use the other inequality $H^{**} \leq H$  in ~\eqref{h_bounds} and combine as above with the Hopf-Lax formula. The upper-bound:
        \[
    \min_{y\in\mathbb{R}^m}  \bigg \{t H^{*}\left(\frac{x-y}{t}\right) + y_{(2)} \bigg\}  \geq V(x,t) ,
    \]
    then follows from the generic fact in analysis that for any continuous function $f$, $f^{***}=f^{*}$. The bounds on $e^{*}(\mathcal{H})$ are obtained then from \cite{vijayan2025minimaxoptimalalgorithmsactive}'s result $V(0,1)=e^{*}(\mathcal{H})$ and observing $\sup_w \tilde{H}^{*}_w=H^{*}$.   
\end{proof}
\begin{remark}
The $R^{go}_B$ upper bounds of \cite{komiyama2022minimax} was shown by \cite{vijayan2025minimaxoptimalalgorithmsactive} to be just evaluations of the functional $y_{(2)}$ at equilibrium points of a $B$-stage repeated Stackelberg game. We believe that for any such Stackelberg equilibrium point $y^{*}_B$, it is the case that $H^{*}(-y^{*}_B) \leq 0$. This would imply that the upper bounds of Theorem~\ref{Thm3} improve on the $R^{go}_B$ of \cite{komiyama2022minimax}. We are unable to show this and leave it for future work. 
\end{remark}
%\rred{\section{To Do}
%\begin{itemize}
    % \item Prove that this does not give some trivial upper/lower bound.
    % \item Check the growth condition.
    % \item Optimize over all possible \textbf{convex + concave} decompositions.
    % \item Summarize bounds from Minimax
    % \item Add corollary optimization over all convex functions. 
    %\item Optimize for L2, add it if possible. Else show it more a few examples
    % \item Hamiltonian relaxation
%\end{itemize}}

\section{Conclusion.}
We obtained flexible upper and lower bounds to the minimax-optimal error exponent $e^{*}(\mathcal{H})$, by utilizing the PDE characterization~\eqref{hji_pde} along with comparison results for viscosity solutions on convex or concave approximations to Hamiltonian and initial conditions of the PDE. The bounds obtained are more general than prior work and are flexibly parametrized by convex functions, sets and lower Hamiltonian relaxations. We conjecture that these bounds are strict improvements over the existing work, like \cite{komiyama2022minimax,vijayan2025minimaxoptimalalgorithmsactive}.
\bibliography{arxiv}
\bibliographystyle{alpha}
\clearpage

\thispagestyle{empty}
%\begin{appendices}
%\end{appendices}
\end{document}